\documentclass{article}[11pt]
\usepackage{graphicx} 
\usepackage[english]{babel}
\usepackage[UKenglish]{isodate}
\usepackage[letterpaper,top=2cm,bottom=2cm,left=3cm,right=3cm,marginparwidth=1.75cm]{geometry}
\usepackage[bottom]{footmisc}

\usepackage{amsmath,amssymb,bm, amsthm, comment}
\usepackage{graphicx,float }
\usepackage[colorlinks=true, allcolors=blue]{hyperref}
\usepackage{forest}

\usepackage{setspace}
\usepackage[backend=bibtex,style=authoryear, sorting=nyt]{biblatex}
\newcommand{\mE}{\mathbb{E}}
\newcommand{\mR}{\mathbb{R}}
\newcommand{\mP}{\mathbb{P}}

\DeclareMathOperator*{\argmin}{argmin}

\DeclareMathOperator*{\argmax}{argmax}

\newtheoremstyle{break}
{\topsep}{\topsep}%
{\normalshape}{}%
{\bfseries}{}%
{\newline}{}%

\theoremstyle{break}
\newtheorem{thm}{Theorem}
\newtheorem{lem}[thm]{Lemma}

\theoremstyle{break}
\newtheorem{defn}[thm]{Definition}
\newtheorem{ass}[thm]{Assumption}

\usepackage[ruled, lined, linesnumbered, commentsnumbered, longend]{algorithm2e}

\title{\Large Bayesian Analysis of High Dimensional Vector Error Correction Model}
\author{Parley Ruogu Yang\footnote{University of Cambridge, United Kingdom} \ 
 and Alexander Y. Shestopaloff\footnote{Queen Mary University of London, United Kingdom and Memorial University of Newfoundland,
Canada}}
\date{This version: \today\\Most recent version: \url{
		https://shorturl.at/tTXY9}}

\begin{document}

\maketitle

\begin{abstract}

Vector Error Correction Model (VECM) is a classic method to analyse cointegration relationships amongst multivariate non-stationary time series. In this paper, we focus on high dimensional setting and seek for sample-size-efficient methodology to determine the level of cointegration. Our investigation centres at a Bayesian approach to analyse the cointegration matrix, henceforth determining the cointegration rank. We design two algorithms and implement them on simulated examples, yielding promising results particularly when dealing with high number of variables and relatively low number of observations. Furthermore, we extend this methodology to empirically investigate the constituents of the S\&P 500 index, where low-volatility portfolios can be found during both in-sample training and out-of-sample testing periods.

\end{abstract}

\pagebreak

\section{Introduction}
Let $Y$ denote a $p$-dimensional multivariate time series, with $Y_t \in \mR^p$ being observed at time $t \in [T]$, and $Y_{t,j} \in \mR$ denoting observation of  $j$-th variable at time $t$. We write $Y(\cdot,j)\in\mR^{1\times T}$ to denote the $j$-th variable of vector $Y$, which is a univariate time series. Stationarity of $Y(\cdot,j)$ is considered as the second-order stationarity. Precisely, we say $Y(\cdot,j)$ is stationary if the first moment being finite and independent of time and that the auto-covariance function being finite and independent of time. We denote $Y(\cdot,j) \sim I(0)$ in this scenario. When $Y(\cdot,j) \sim I(0)$ for all $j \in [p]$, we write $Y\sim I(0)$. When $Y(\cdot,j)$ is non-stationary but when $\Delta Y(\cdot,j) $ defined by $\Delta Y(t,j):=Y(t,j)-Y(t-1,j)$ is stationary, we write $Y(\cdot,j) \sim I(1)$. Particularly, when $Y(\cdot,j)\sim I(1)$ for all $j\in[p]$, we write $Y \sim I(1)$. 

In this paper, we focus on non-stationary time series $Y\sim I(1)$ and centre the investigation on a Bayesian method to analyse Vector Error Correction Model (VECM), which contains a low-rank matrix $\Pi \in \mR^{p\times p}, rank(\Pi)=r$ that lead to linearly independent cointegration relationships, represented by cointegration vectors $\beta_1,...,\beta_r$ such that $\beta_j Y(\cdot,j) \sim I(0)$. We employ Bayesian modelling on the potentially-sparse $\Pi$ and further the Bayesian inference on such method of estimations. Two algorithms are designed and run on simulated examples, which produce good results under high $p$ and relatively low $T$. We further this methodology to empirically study the constituents of S\&P 500 index, and found low-volatility portfolios in both in-sample training and out-of-sample testing periods.

\subsection{VECM and a motivation for High Dimensional VECM}
VECM  was first proposed by \textcite{EG1987}, followed by sequential Johansen tests \parencite{SJ91} on determining the number of cointegration vector. The general form of VECM takes the form of \begin{equation}	\label{GeneralVAR}
	\Delta Y_t = \Pi Y_{t-1} + f(X_t) + \varepsilon_t, \ \ \varepsilon_t\sim N(0, \Sigma) \text{ iid }
\end{equation}
where $X \sim I(0), Y\sim I(1)$ and $f$ is linear. $\Pi \in \mR^{p \times p}$ is known as the cointegration matrix.  In this paper, we focus on the VAR(1) part of the VECM, which dismisses the $X$ part of the \autoref{GeneralVAR}. Precisely, we investigate \footnote{The extra terms may often be handled by Partial Least Squares projection, which then reduce the estimation equation back to the VAR(1) form which we focus on.} \begin{equation}	\label{VECM}
	\Delta Y_t = \Pi Y_{t-1} + \varepsilon_t, \ \ \varepsilon_t\sim N(0, \Sigma) \text{ iid, } rank(\Pi) = r \in \{0,1,...,p-1\}
\end{equation}

The idea of high dimensional VECM falls in the trendy literature surrounding `big data', in which we have a large amount of variables of interest ($p$) and a substantial amount of observations taken over time ($T$). In a traditional statistical setting, we require more observations than unknown parameters --- in the setting of \autoref{VECM} we need $T>p^2 + \frac{p(p+1)}{2}$. However, that often means we need a greater amount of observations than the data actually have, in order to estimate the equation. This is particularly concerning in light of change-points or other potential changes of the data over time, which naturally constrains the size of $T$. High dimensional VECM is therefore motivated in the regime where $p, T$ are large but $T\leq p^2$. Certainly, for estimating the non-sparse coordinates of the matrix $\Pi$, it is reasonable to impose a feasibility assumption where $T>pr$.

A relevant example that further motivates high dimensional VECM goes as follows. Consider $p$ number of stocks in an asset allocation exercise, i.e. let $Y_t$ be a $p$-dimensional vector of stock price at time $t$, encompassing $p$ number of stocks. It is common to assume $Y \sim I(1)$. Then, if we find a non-zero vector $\beta \in \mR^{1 \times p}$ such that $\beta Y \sim I(0)$, we have $\mE[\beta Y_0] = \mE[\beta Y_t] \forall t$. This has profound implications in the portfolio construction, as $\beta$ would be considered as a stable portfolio, i.e. the value of which is expected to be unchanged over time. We further investigate this empirically in section 4.

\subsection{Literature Review}

As to literature review, we first review on the method proposed by \textcite{LiangSchienle19} for Partial Least Squares (PLS), then proceeds into the literature of Spike-and-Slab Lasso (SSL), which was first introduced by \textcite{Rockova18}, and subsequently extended in \textcite{RG2018}. 

Write $A_t := \Delta Y_t := Y_t - Y_{t-1}$ and $B_t := Y_{t-1}$ and let $A,B \in \mR^{p\times T}$ be constructed by $A=[A_1, ..., A_T]$ and $B=[B_1, ..., B_T]$. The PLS estimation of $\tilde{\Pi}$ proposed by \textcite{LiangSchienle19} takes the form of \begin{equation}
	\tilde{\Pi} = (AB^\intercal)(BB^\intercal)^{-1} \label{tildePi}
\end{equation}A QR decomposition is then proposed, taking the form of  \begin{equation}
\tilde{\Pi} = \tilde{R}^\intercal \tilde{S}^\intercal  \label{PiQR}
\end{equation}
 where $\tilde{R}$ is a p-dimensional upper triangular matrix and $\tilde{S} \in  \mR^{p \times p}$ is orthogonal.  \textcite{LiangSchienle19} proceed with group-lasso regression to find a low-rank $\hat{R}$ based on $\tilde{R}$, however, we opt for a different route to investigate the true rank of $\Pi$. An important takeaway from this method is that the determination of rank $r$ can be written as $\tilde{r}=rank(\tilde{R})=rank(\tilde{\Pi})$, and subsequently $\hat{r}=rank(\hat{R})=rank(\hat{R} \tilde{S}^\intercal)$ for an alternative estimator $\hat{R}$ of the true $R$. The rank collected is purely the number of non-zero columns of $R$, as all columns of $R$ are linearly independent, due to the nature of the QR decomposition (such as via Gram–Schmidt process).

The SSL literature mainly concerns scalar response multivariate-regression, i.e. the type taking the form of $y_i=x_i^\intercal \beta + \varepsilon_i$ for each index $i$, with $x_i$ and $\beta$ being $p-$dimensional vectors and $y_i$ being an one-dimensional observation, while $\varepsilon_i$ is an iid draw from a Gaussian distribution with zero mean and unknown variance. $\beta_0$ is denoted as the true value of the $\beta$ parameter. SSL can be thought of as a mixing distribution being imposed on the parameter $\beta$, conditional on certain hierarchical Bayes parameters $\gamma$, which is then conditionally Bernoulli based on $\theta$, which denotes prior expected fraction of non-zero entries of the true $\beta_0$. A posterior mode estimator then aims to maximise the posterior distribution $\pi(\beta|\theta)$. Write $\hat{\beta}=\argmax \pi(\beta|\theta)$. In \textcite{RG2016}, $\theta\sim \cal B$ (a Beta distribution on $\theta$) was proposed to enable further inference on the distribution of the estimator conditional on $\theta$, denoted $\pi(\hat{\beta}|\theta)$, and thereafter the conditional expectation $\mE[\theta|\hat{\beta}]$. In \textcite{RG2018}, further analyses were drawn on the asymptotic behaviour of $\hat{\beta}$ such as bounding the limiting posterior probability based on the true $\beta=\beta_0$ with data $D$ and $\theta$, namely $\sup_{\beta_0} \mE_{\beta_0}\left[\mP_{\cdot|D,\theta}[\hat{\beta} : || \hat{\beta} - \beta_0 ||_1 > \varepsilon] \right]$ with $\varepsilon$ depending on 
the number of observation $n$ and features dimension $p$. \textcite{CV2012} is another key source of literature that provide asymptotic inference on this type of Bayesian system.

In our study, we alter the formation of SSL distribution into a row-wise sparsity identification framework towards a $p\times p$ dimensional matrix which is believed to be potentially sparse. This alteration enables further inference and algorithmic interaction towards the rank of $\Pi$, which is the order of cointegration. Such alteration may still be studied asymptotically, as we further the asymptotic investigation into statements on the expected probabilities of $\hat{r}=r$, i.e. the event where we correctly estimate the rank.

\subsection{Organisation}
The rest of this paper is as follows. Section 2 gives details on SSL distribution on a decomposed VECM, followed by the associated asymptotic analysis and algorithmic design. Section 3 then reports simulated examples under various sample size ($T$) and features dimension ($p$), followed by empirical application to stable portfolio in section 4.

\section{Methodology}

\subsection{Simple SSL distribution on decomposed VECM}

We first introduce the SSL distribution, in both univariate and multivariate contexts, then apply it to the VECM decomposition.

\begin{defn}[SSL distribution]
	Let $X\in\mR$ be a random variable. A Spike-and-Slab Lasso (SSL) pior distribution on $X$ is defined as 
	\begin{align}
		\pi(X|\gamma) :=& (1-\gamma)\psi_0(X) + \gamma \psi_1(X) , \ \ \gamma \in [0,1] \\
		\psi_j(x) :=& \frac{\lambda_j}{2}\exp(-\lambda_j |x|) \ \ \forall j \in \{0,1\} \text{ with } \lambda_0> \lambda_1 >0
	\end{align}
We write $X|\gamma \sim SSL(\gamma, \lambda_0, \lambda_1)$. Let $Y\in\mR^p$ be a random variable. We write $Y|\gamma \sim SSL(\gamma, \lambda_0, \lambda_1)$ if for all entries $j \in [p]$, $Y_j|\gamma \sim SSL(\gamma, \lambda_0, \lambda_1) \ $iid
\end{defn}
We may interpret $X|\gamma$ as a mixture of two Laplace distribution, with $\psi_0$ being spiky at around 0 and $\psi_1$ less spiky.  $Y|\gamma$ being henceforth a vectorised version of such a distribution, with identical probabilities being attributed.

Recall the \autoref{tildePi} and \autoref{PiQR} composition of the pre-estimate of $\Pi$. In what follows, we attribute an SSL distribution on the decomposed cointegration matrix in VECM. We use the notation $A_t:=\Delta Y_t$ and $B_t:= \tilde{S}^\intercal Y_{t-1}$, with centralisation taken place from $A_t$ to $\tilde{A}_t$ and normalisation from $B_t$ to $ \tilde{B}_t$ (see appendix \ref{DP} for granular details of data processing).

\begin{defn}[Simple SSL distribution on decomposed VECM]\label{SSL-VECM0}
	Let $\tilde{A}_t, \tilde{B}_t \in \mR^{p}$ for all $t\in [T]$. Let $\theta\in(0,1)^p$ be a fixed hyperparameter.  A simple SSL distribution on decomposed VECM takes the form of 
		\begin{align}
		\tilde{A}_t|\tilde{B}_t, R, \sigma^2  &\sim N(R^\intercal \tilde{B}_t, \sigma^2 I_p) , 	\ \	\pi(\sigma^2) \propto \sigma^{-2} & \forall t\in [T] \label{Gauss_AB} \\
		R(\cdot,j) | \gamma_j &\sim  SSL(\gamma_j, \lambda_{0,j}, \lambda_{1}), \ \ \gamma_j |\theta_j \sim Bernoulli(\theta_j)  & \forall j \in [p] \label{SSL_R} 
	\end{align}
\end{defn}

In the above model, we have data $\{A_t,B_t\}_{t\in[T]}$ and modelling parameters $R, \sigma$ while $\gamma$ can be marginalised due to their conjugate distributions. As a notational remark: the $\psi_{0,j}(\cdot)$ and $\psi_1(\cdot)$ functions are defined accordingly to the $\lambda_{0,j}, \lambda_1$ as part of the SSL distribution. 

The above model is not the main model we interact with, as $\theta$ is treated as a hyperparameter --- in the later section we generalise this into a random variable and derive asymptotic analysis on the estimator. 
For the purpose of intuition, we draw attention to the Bayes Estimator of the simple model as per Definition\autoref{SSL-VECM0}, which has a cleaner form of solution as $\theta$ is treated as a fixed hyperparameter. We formally define the Bayesian Estimator below, which is also known as `global posterior mode' or `maximum a posteriori' in Bayesian literature (such as \cites{barber2011bayesian}{Rockova18}).

\begin{defn}[Bayes Estimator]
	Let $D$ be a collection of time series data with initialisation $D_0$ and $\phi$ be modelling parameters living in the space $\Phi$. If the posterior $\pi(\phi|D, D_0)$ is well defined over $\phi \in \Phi$ and it attains its unique maximiser, then the Bayes Estimator is the maximiser \begin{equation*}
		\hat{\phi}^{BE} := \argmax_{\phi \in \Phi} \pi(\phi|D,D_0)
	\end{equation*}
\end{defn}

Under the above setting, we note that $\phi=(R,\sigma)$ and $\Phi = \mR^{p\times p} \times \mR_{> 0}$. We can derive the Bayes Estimator into the followings: define $e_t :\mR^{p\times p} \to \mR^p$ by $e_t(R) =\tilde{A}_t - R^\intercal \tilde{B}_t$ for all $t$,
\begin{align}
	(\hat{R}^{BE}, \hat{\sigma}^{BE}) =& \argmin_{R\in \mR^{p\times p} 
		, \sigma \in \mR_{> 0}} \sum_{t\in[T]} \frac{||{e_t}(R)||_2^2}{2\sigma^2} + (pT+2) \log(\sigma) + \text{pen}(R,\theta) \label{R_BE} \\
	\text{pen}(R,\theta) := & \sum_{j \in [p]} \left( \lambda_1 ||R(\cdot,j)||_1 + \sum_{i\in [p]} \log \left( p_j^*(R(i,j),\theta_{j}) \right) \right) \label{pen} \\
	p_j^*(x,\theta_{j}) :=& \frac{\theta_{j} \psi_{1}(x)}{\theta_{j} \psi_{1}(x) + (1-\theta_{j}) \psi_{0,j}(x)}
\end{align}

Derivation details are at appendix \ref{DBE}. In essence, the first two terms of \autoref{R_BE} are due to Gaussian assumption (\autoref{Gauss_AB}) with Jeffery prior on the variance. The penalty term $\text{pen}(R,\theta)$ is due to the SSL assumption (\autoref{SSL_R}). 

The consequential estimator takes the form of a penalised least square where a $||\cdot||_1$ norm is applied to the entire matrix $R$, together with some extra terms measured by $p_j^*(\cdot,\cdot)$.

\subsection{Hierarchical Bayes Modelling of SSL on VECM}\label{HBA}	
The full model utilises the simple SSL distribution as the core, but addes a further distribution on $\theta$ to enrich the hierarchical architecture. Furthermore, the Gaussian assumption of the data is generalised by replacing \autoref{Gauss_AB} with the conditional variance noted as $\sigma_j^2$ for each entry $\tilde{A}_{t,j}$. Formal definition are set as below.

\begin{defn}[Hierarchical Bayes Model of SSL distribution on decomposed VECM]\label{SSL-VECM1}
	Let $\tilde{A}_t, \tilde{B}_t \in \mR^{p}$ for all $t\in [T]$.  Let $a,b \in \mR_{> 0}^p$ be the parameters of prior distribution of $\theta$. A Hierarchical Bayes Model of SSL distribution on decomposed VECM takes the form of 
	\begin{align}
		\tilde{A}_{t,j}|\tilde{B}_t, R, \sigma_j^2  &\sim N(R(\cdot,j) \tilde{B}_t, \sigma_j^2) , 	\ \	\pi(\sigma_j^2) \propto \sigma_j^{-2} & \forall t\in [T],  j  \in [p] \label{Gauss_AB1} \\
		R(\cdot,j) | \gamma_j &\sim  SSL(\gamma_j, \lambda_{0,j}, \lambda_{1}), \ \ \gamma_j |\theta_j \sim Bernoulli(\theta_j), \ \ \ \theta_j  \sim {\cal B} (a_j,b_j)  & \forall j \in [p] \label{SSL_R1} 
	\end{align}
\end{defn}

The Bayes Estimator of this alternative model is derived in appendix \ref{DBE2}, which takes the following form: define vector ${\epsilon}(j)\in\mR^{T}$ by $({\epsilon}(j))_t = \frac{(e_t(R))_j}{\sqrt{2}\sigma_j}$ for all entry $j$, then
\begin{equation}
	(\hat{R}^{BE}, (\hat{\sigma}_j^{BE})_{j\in[p]}) = \argmin_{R\in \mR^{p\times p} 
		, ({\sigma}_j)_{j\in[p]} \in \mR^p_{> 0}} \sum_{j\in[p]} \left( {||{\epsilon}(j)||_2^2} + (T+2) \log(\sigma_j) + \lambda_1    ||R(\cdot,j)||_1 - \log \left(p^{**}_j(R,\theta)  \right)   \right) 	\label{R_BE2}
\end{equation}

with $p^{**}_j(\cdot,\cdot)$ being defined as 
\begin{equation}
	p^{**}_j(R,\theta) := \int_{(0,1)}   \prod_{i\in[p]} \frac{\theta_j\pi(\theta_j)}{p_j^*(R(i,j),\theta_j)} d\theta_j
\end{equation}

Further manipulation of the function $p^{**}$ may be extended to interact with, at an approximated level, the posterior $ \mE[\theta_j | \hat{R}^{BE}( \cdot,  j )]$. See \textcite{RG2018} for further discussions.

	\subsection{Asymptotic Analysis}
    Write $\beta=R^\intercal$. The idea of asymptotic analysis is to provide  probabilistic inference towards the limiting behaviour of the estimation matrix $\hat{R}^{BE}$ (equivalently thereafter $\hat{\beta}$) in \autoref{R_BE2} and the rank $\hat{r}$, which is defined as the number of non-zero columns of $\hat{R}^{BE}$. In what follows, we introduce some approximate notions of $\hat{\beta}$ and the associated ranks to enable the adaption towards some existing theories and our proposed asymptotic statements.
		

	\begin{defn}
		[Intersection line, generalised dimensionality and rank]
		
		Let $\delta_i:=\frac{ \log((p_i^*(0))^{-1}-1) }{\lambda_{0,i}-\lambda_1}$. This is known as the intersection line.
		
	We construct matrix $\tilde{\beta}$ entry-by-entry as follows:\begin{equation*} \forall i,j, \ \
			\tilde{\beta}_{i,j}=
			\begin{cases}
				\beta_{i,j} & \text{if }  |\beta_{i,j}| >\delta_i \\
			 	0	& \text{else}
			\end{cases}
		\end{equation*} We define the generalised dimensionality of row $\tilde{\beta}_i:= \tilde{\beta}(i,\cdot)$ and the matrix as \begin{equation}
			\gamma_i(\beta) = || \tilde{\beta}_i  ||_0, \ \ \gamma(\beta) = \sum_{i \in [p]}\gamma_i(\beta)
		\end{equation}
		We also define the generalised rank $\tilde{r}(\beta_i):= | \{i \in [p]: \gamma_i(\beta) >0 \} | $

		\end{defn}
	
	As a notational remark, note that for any estimate $\hat{\beta}$, we have the associate notions of $	\tilde{\beta}(\hat{\beta}),	\gamma_i(\hat{\beta}), \gamma(\hat{\beta}), \tilde{r}(\hat{\beta})$. For the true $\beta$, we write $q=||\beta||_0=\gamma(\beta)$. Let $r$ be the  rank of the true $\beta$.

	\begin{ass}
		[Asymptotic assumptions]
		The following asymptotic assumptions are required for Lemma\autoref{thm2} and Theorem\autoref{keythm2}: there exists $a, b \geq 2$ where $\frac{1-\theta}{\theta} \sim p^a$ and $\lambda_{0,i} \succsim p^b \ \forall i $. $\lambda_1$ satisfies $\lambda_1 \in [\frac{\sqrt{T}}{p}, 4 \sqrt{T \log p}]$
	\end{ass}
	
	Below, we briefly outline the posterior concentration, mainly from \textcite{RG2018}, which	leads to the proof of Theorem\autoref{keythm2}. The notation $\mP_{\cdot | D, \theta }(\cdot)$ denotes the posterior concentration given data $D$ and true parameter $\theta$.
	
	\begin{lem}[Posterior concentration under \textcite{RG2018}] \label{thm2}
	Write $q_i:=||\beta_i||_0$ and let $q^M:=\max_{i \in [p]} q_i$.
	
	There exists constants $M_i>0$ for all $i \in [p]$, and $K>0$ decreasing with $p$ such that as $p, T \to \infty$,
	\begin{align}
	\forall i, \ \	\sup_{\beta_i} \mE_{\beta_i} \left[\mP_{\cdot | D, \theta }\left(\hat{\beta}_i:  ||\beta_i-\hat{\beta}_i||_1 > q_i M_i \sqrt{\frac{\log p}{T}}
		\right) \right] & \to 0 \label{thm1_0} \\
		\sup_{\beta} \mE_\beta [\mP_{\cdot | D, \theta }(\hat{\beta}: \gamma(\hat{\beta})> r q^M(1+K))] & \to 0 \label{thm1}
	\end{align}
	\end{lem}

\begin{proof}
The first limit (\autoref{thm1_0}) is exactly the same as what has been proved in \textcite[Theorem 8]{RG2018}, as in our setting, $\beta_j | \gamma_j$ is SSL-distributed.

The second limit (\autoref{thm1}) can be proved with the help of Theorem 7 in \textcite{RG2018}, which gives a row-level sparsity criteria as follows:
there exists $K_i>0$ decreasing with $p$ such that \begin{equation}
\sup_{\beta_i} \mE_{\beta_i} [\mP_{\cdot | D, \theta }(\hat{\beta}_i:  \gamma_i(\hat{\beta}) > q(1+K_i))] \to 0  \ \ \text{ as } p, T \to \infty\label{thm1_p1}
\end{equation}
 Consider the set $I:=\{i \in [p]: \beta_i = 0\}$. Then by \autoref{thm1_0}, $\sup_{\beta_i} \mE_{\beta_i} \left[\mP_{\cdot | D, \theta }\left(\hat{\beta}_i:  \hat{\beta}_i \neq 0
\right) \right]  \to 0$. Therefore \begin{equation}
\sup_{\beta} \mE_{\beta} \left[\mP_{\cdot | D, \theta }\left(\hat{\beta}:  \gamma_i(\hat{\beta}) > 0
\right) \right]  \to 0 \ \ \ \forall i \in I \label{thm1_p2}
\end{equation}

For $i \notin I$, we pick $K = \max_{i \in [p]\setminus I} K_i$. Because of \autoref{thm1_p1}, we have \begin{equation}
\sup_{\beta} \mE_{\beta} \left[\mP_{\cdot | D, \theta }\left(\hat{\beta}: \sum_{i \in [p] \setminus I} \gamma_i(\hat{\beta}) > \sum_{i\in [p] \setminus I} q(1+K_i) \right) \right] \to 0  \label{thm1_p3}
\end{equation}
By the construction of $r$, we have $|[p]\setminus I| = r$ therefore $rq^M(1+K)\geq \sum_{i\in [p] \setminus I} q(1+K_i) $. This combines with \autoref{thm1_p2} and \autoref{thm1_p3} leads to \autoref{thm1}.
\end{proof}


	\begin{thm}[Rank-consistent posterior concentration] \label{keythm2}
		Let $r, q^M$ be independent of $p$ and let $\delta_i, p,T$ satisfy $\frac{\log p}{T} \geq  \delta_i^2 (qM)^{-2} $ eventually as $p,T\to\infty $, for all $i\in[p]$. Then 
			$$\sup_{\beta} \mE_\beta [\mP_{\cdot | D, \theta }(\hat{\beta}: \tilde{r}(\hat{\beta}) \neq r)] \to 0 $$ as $p, T \to \infty$
	\end{thm}

	\begin{proof}
		Write $q^\prime:= r q^M(1+K)$. Let $A:=\{\hat{\beta}: \tilde{r}(\hat{\beta}) \neq r \}$, $B_j:=  \{\hat{\beta}: \gamma(\hat{\beta}) = j \} \forall j \leq  q^\prime$ and $C:=  \{\hat{\beta}: \gamma(\hat{\beta}) > q^\prime \} $. 
		Observe $$
		\sup_{\beta} \mE_\beta [\mP_{\cdot | D, \theta }(A)]\leq \sum_{j=0}^{ q^\prime} \sup_{\beta} \mE_\beta [\mP_{\cdot | D, \theta }(B_j \cap A )] + \sup_{\beta} \mE_\beta [\mP_{\cdot | D, \theta }(C \cap A )] 
		$$
		
		Now, it is clear from \autoref{thm1} of lemma\autoref{thm2} that $\sup_{\beta} \mE_\beta [\mP_{\cdot | D, \theta }(C\cap A)] \leq \sup_{\beta} \mE_\beta [\mP_{\cdot | D, \theta }(C)] \to 0$ so it remains to prove $\sup_{\beta} \mP_{\cdot | D, \theta }(B_j \cap A) \to 0 $ for all $j\leq q^\prime$.
		
		\begin{enumerate}
			\item When $j\leq q^\prime$ and $j \neq q$, consider $\hat{\beta} \in B_j$. Observe that there exists $i$ such that $\gamma_i(\hat{\beta}) \neq ||\beta_i||_0$. Henceforth by construction, $||\tilde{\beta}_i - \beta_i||_1 > \delta$. Now, let set $I=\{ i : \gamma_i(\hat{\beta}) \neq ||\beta_i||_0 \}$, where cardinality $|I| \leq q+q^\prime$ as $j \leq q^\prime$. Observe \begin{equation}\label{keytoprove}
				\mP_{\cdot | D, \theta } (B_j \cap A) \leq \mP_{\cdot | D, \theta } (B_j )\leq (q+q^\prime) \max_{i \in I}  \mP_{\cdot | D, \theta }(   \{\hat{\beta}_i: ||\tilde{\beta}_i - \beta_i||_1 > \delta_i \} )
			\end{equation}
			
			Now notice, for all $i\in I$, by construction, $\tilde{\beta_i}=\hat{\beta_i}$, and so by lemma\autoref{thm2}, we know $ \mP_{\cdot | D, \theta }( \{\hat{\beta}_i: ||\tilde{\beta}_i - \beta_i||_1 > \delta_i \} ) \to 0$ which completes the proof.
			
			\item When $j=q$, reconsider $\hat{\beta} \in B_q \cap A$. As $ \tilde{r}(\hat{\beta})  \neq r$, the set $I \neq \emptyset$. We can therefore bound, similarly to \autoref{keytoprove}, that \begin{equation}
				\mP_{\cdot | D, \theta } (B_q \cap A) \leq(q+q^\prime)  \max_{i \in I}  \mP_{\cdot | D, \theta }( \{\hat{\beta}_i: ||\tilde{\beta}_i - \beta_i||_1 > \delta_i \} ) \to 0
			\end{equation}
			
		\end{enumerate}
	\end{proof}
Theorem\autoref{keythm2} offers some level of assurance in terms of consistency in rank estimation. In essence, the statement gives a strong probability convergence, globally over $\beta \in \mR^{p\times p}$, of the generalised rank $\tilde{r}$ to $r$. 
	
	\pagebreak
	
\subsection{Algorithms}

In general, the algorithms employ EM techniques (see \textcite{RRG2021} for further details on the EM techniques) to numerically approximate the estimators $(\hat{R}^{BE}, \hat{\sigma}^{BE})$ as per \autoref{R_BE2}. Write these as $(\hat{R}, \hat{\sigma})$ and the approximated rank $\hat{r}:=rank(\hat{R})$. The two algorithms differ by their approaches in $\lambda_0$ searching, in that algorithm 1 assumes the same value of $\lambda_{0,j}= \lambda_0 \forall j $ whereas algorithm 2 allows some degree of different increments of $\lambda_{0,j}$ and associated termination. Such difference is propagated through a uniform sampling on non-sparse coordinates.

\subsubsection{Algorithm 1}
	As for notation, we use $e_{t,j}(k)=  A_{t,j}-\beta_j(k)B_t= A_{t,j}-R(\cdot,j)(k)B_t$ to indicate the fitted $j$-th residual from the associated $k$-th iteration of the matrix estimate $R(k)$ for observations $A_t, B_t$. It is convenient to also use the following functions: $p_\theta^*$ and $\lambda_\theta^*$, both of which are $\mR$ to $\mR_{>0}$ mappings.
	\begin{align}
		p_\theta^*(x) & = \left( 1 + \frac{\lambda_0 (1-\theta)}{\lambda_1 \theta} \times \exp\left( |x| (\lambda_1-\lambda_0) \right) \right)^{-1} \\
		\lambda_\theta^*(x) & = \lambda_0 + (\lambda_1 - \lambda_0)     p_\theta^*(x)
	\end{align}
	
	In \autoref{RwSSLEM}, we initialise a dummy estimate $\hat{r}=p$ and iteratively increase $\lambda_0$ per EM iteration. After each of the EM iterations, we re-estimate the rank. The iteration stops when $\hat{r}$ gives the same rank for $n_r$ times. This is done by the imposition of vector $\tilde{r}$.

	\subsubsection{Algorithm 2}
	In \autoref{RwSSLEM2b}, we first use algorithm 1 to bring down the estimated rank to certain level of sparsity, here we state $\hat{r} < \frac{p}{2}$, but it may vary, certainly for large $p$ to some level such as $\hat{r} < \sqrt{p}$ or $\hat{r} < \frac{T}{p}$. After brining down the estimated rank, we implement an alternative approach where we sample $j\sim U(\{j : \hat{R}(\cdot,j) \neq  \boldsymbol{0}_p\})$, i.e. a coordinate sampled uniformly from the set of non-zero coordinates of the estimated rows. We then implement the EM iterations on this row only, and update rank estimate accordingly. The technical change would also add on a different notion of  $p_\theta, \lambda_\theta$ due to the different $\lambda_{0,j}$ across coordinates. These are re-defined below:
	\begin{align}
	p_\theta^*(x,j) & = \left( 1 + \frac{\lambda_{0,j} (1-\theta)}{\lambda_1 \theta} \times \exp\left( |x| (\lambda_1-\lambda_{0,j}) \right) \right)^{-1} \\
	\lambda_\theta^*(x,j) & = \lambda_{0,j} + (\lambda_1 - \lambda_{0,j})     p_\theta^*(x,j)
	\end{align}
	
	As the randomness (due to uniform sampling) exists, we may ensemble the results and provide a probabilistic output instead of a scalar estimate. For a set of seeds $S$, we may have estimates $\hat{r}(s)$ for $s\in S$. Statistics of these ranks, e.g. median and mean may offer insights into the approximate value of the estimated ranks.

	
	\begin{center}
		\begin{algorithm}[H]
			\caption{Column-wise SSL through EM}
			\label{RwSSLEM}
			\KwIn{Data A,B ;
				\\ Penalisation parameters $\lambda_0, \lambda_1, \Delta_\lambda$ ;
				\\ Initialisers $\theta(0), \beta(0), (\sigma_j^2(0))_{j\in[p]}$; 
				\\ Halting controls $K, \epsilon_{cap}, n_r$; 
				\\ Parametric specifications for prior on $\theta_j$: $a_j$, $b_j$}
			\KwOut{Estimated ranks $ \hat{r}$ and matrix $\hat{R}$}
			Initialise estimation vector $\tilde{r}=(p+1)\times \mathbf{1}_{n_r}$  and estimate $\hat{r}=p$.
			
			Process data $A, B$ such that all rows are centered at zero and all rows of $B$ have variance $T$.
			
			Iterative process on rank determination, by increase of $\lambda_0$:
			
			\While{$\tilde{r} \neq \hat{r} \times  \mathbf{1}_{n_r} $}{
				
				Update $\lambda_0$ with $\lambda_0 + \Delta_\lambda$
				
				EM update on $(\beta, \theta, \sigma^2)$:
				
				\For{$j\in[p]$}{
					Initialise $\theta_j(0), \beta_j(0), \sigma_j^2(0)$, then
					
					\For{$k\in[K]$}{
						\begin{align*}
							\beta_j(k) = & \argmin_{\beta_j \in \mR^p} \left\{ \sum_{t\in [T]} e_{t,j}(k)^2 + 2 \sum_{i \in [p]} \sigma_j^2(k-1) \lambda_{\theta_j(k-1)}^*(\beta_{i,j}(k-1)) |\beta_{i,j}(k)| \right\} \\
							\theta_j(k) = &  \frac{a_j-1+ \sum_{i\in[p]} p_{\theta_j(k-1)}^* (\beta_{i,j}(k)) }{a_j+b_j+p-2} \\
							\sigma_j^2(k) =& \frac{\sum_{t\in [T]} e_{t,j}(k)^2}{T-2} \\
							\epsilon(k) =& ||\beta_j(k)-\beta_{j}(k-1)||_2 
						\end{align*}
					\If{$\epsilon(k) < \epsilon_{cap}$}{
							break the $k$ for-loop
						}
					}
					Collect $\hat{\beta_j} = \beta_j(k)$
					
				}
				
				Collect $\hat{r}=rank\left(\left[\hat{\beta_1} | ... | \hat{\beta_p}\right]\right)$
				
				Replace $\tilde{r}_{1:(n_r-1)}$ with $\tilde{r}_{2:n_r}$ and replace $\tilde{r}_{n_r}$ with $ \hat{r}$
			}
			Output estimated rank  $\hat{r}$ with the final matrix $\hat{R} = \left[\hat{\beta_1} | ... | \hat{\beta_p}\right]$
		\end{algorithm}
		
	\end{center}
	
	\nopagebreak[1]
	\begin{center}
		\begin{algorithm}[H]
			\caption{Column-wise SSL through EM with entry-wise $\lambda_0$ increment}
			\label{RwSSLEM2b}
			\KwIn{Data A,B ; Seed for uniform sampling;
				\\ Penalisation parameters $\lambda_0 = (\lambda_{0,j})_{j \in [p]}, \lambda_1, \Delta_\lambda, \delta_\lambda$ ;
				\\ Initialisers $\theta(0), \beta(0),  (\sigma_j^2(0))_{j\in[p]}$; 
				\\ Halting controls $K, \epsilon_{cap}, n_r$; 
				\\ Parametric specifications for prior on $\theta_j$: $a_j$, $b_j$
			}
			\KwOut{Estimated ranks $ \hat{r}$, matrix $\hat{R}$}
			
			Process data $A, B$ such that all rows are centered at zero and all rows of $B$ have variance $T$.
			
			Initialise $\hat{r} = p$
			
			\While{$\hat{r}\geq \frac{p}{2}$}{
							
				Use \autoref{RwSSLEM} to update $\lambda_0$ with vectorised $\lambda_0+\Delta_\lambda$ and EM update on $(\beta,\theta,\sigma^2)$ 
				
				Collect $\hat{r}(\lambda_0)=rank\left(\hat{R}(\lambda_0)\right)$

			}
			
			Initialise estimation vector $\tilde{r}= (\hat{r}+1)\times \mathbf{1}_{n_r}$
			
			Initialise set $P=\{j: \hat{R}(\cdot,j) \neq  \boldsymbol{0}_p \}$. Set seed for sampling.
			
			Iterative process on rank determination, by increase of $\lambda_{0,j}$ for random $j$:
			
			\While{$\tilde{r} \neq \hat{r} \times  \mathbf{1}_{n_r} $ and $\hat{r}\neq 0$}{
				
				Sample $j \sim U(P)$
				
				Update $\lambda_{0,j}$ with $\lambda_{0,j} + \delta_\lambda  |P|$
				
				Initialise $\theta_j(0), \beta_j(0), \sigma_j^2(0)$ and EM update on $(\beta_j, \theta_j, \sigma_j^2)$ with the following steps:
				
				\For{$k\in[K]$}{\vspace{-0.2cm}\begin{align*}
						\beta_j(k) = & \argmin_{\beta_j \in \mR^p} \left\{ \sum_{t\in [T]} e_{t,j}(k)^2 + 2 \sum_{i \in [p]} \sigma_j^2(k-1) \lambda_{\theta_j(k-1)}^*(\beta_{i,j}(k-1),j) |\beta_{i,j}(k)| \right\} \\
						\theta_j(k) = &  \frac{a_j-1+ \sum_{i\in[p]} p_{\theta_j(k-1)}^* (\beta_{i,j}(k),j) }{a_j+b_j+p-2} \\
						\sigma_j^2(k) =& \frac{\sum_{t\in [T]} e_{t,j}(k)^2}{T-2} \\
						\epsilon(k) =& ||\beta_j(k)-\beta_{j}(k-1)||_2 
					\end{align*}
					\If{$\epsilon(k) < \epsilon_{cap}$}{
						break the $k$ for-loop
					}
				}
				
				Collect $\hat{\beta_j} = \beta_j(k)$
				
				Obtain $\hat{R}(\lambda_0)$ by copying the previous  $\hat{R}(\lambda_0)$ but replace the $j$-th column with $\hat{\beta}_j$. 
				
				Collect $\hat{r}(\lambda_0)=rank\left(\hat{R}(\lambda_0)\right)$
				
				\If{$\hat{\beta_j} = \boldsymbol{0}_p$}{
					Update $P$ by $P \setminus \{j\}$
				}
				
				Replace $\tilde{r}_{1:(n_r-1)}$ with $\tilde{r}_{2:n_r}$ and replace $\tilde{r}_{n_r}$ with $\hat{r}(\lambda_0)$
			}
			
			Output estimated rank  $\hat{r}$ with the final matrix $\hat{R}$ and paths.

		\end{algorithm}
		
	\end{center}

	
\section{Simulation}
\subsection{Outline of Simulated Examples}
The principle idea of the simulation task is to generate some big VAR(1) process and run algorithms 1 and 2 and collect the estimated ranks. The proposed VAR(1) process takes the following form
\begin{align}
	\Phi := & (J(0,r) \oplus I_{(p-r)} ) + \mathbf{1}_{r, r+1}  \\
	Y_t := & \Phi Y_{t-1} + \varepsilon_t , \ \ \varepsilon_t \sim N(0,\sigma^2I_p) \  \ \ \forall t \in [T] \label{Simu_DGP}
\end{align}
with initialisation $Y_0=0$. $J(0,r)$ refers to the Jordan Block of dimension $r$ with diagonal taking value 0. If we write \autoref{Simu_DGP} in the form of \autoref{VECM}, then $$\Pi = \Phi-I_p = (J(0,r) - I_{r})  \oplus 0_{(p-r)} + \mathbf{1}_{r, r+1} $$
it is easy to see that $rank(\Pi) = r$.

For each tuple $(p,r,T)$, we generate $S=100$ samples all satisfying $Y(s) \sim I(1)$ (checked by ADF test at 5\% significance level), then proceed with the relevant algorithms to collect ranks. As algorithm 2 gives a probabilistic inference on the rank instead of a numerical estimate, we instead enquire algorithm 2 100 times per sample, and take the mean as the estimated rank. The general algorithm is drawn below.

\begin{center}
	\begin{algorithm}
		\caption{General algorithm for simulated examples}
		\label{SimuAlgo}
		\KwIn{$T,p,r$; and number of successfully simulated samples $S=100$
		}
		\KwOut{Estimated ranks $ \hat{r}(i)$ for $i\in I$ with $|I| = S$}
		Initialise $\texttt{successful\_seed}=0$ and seed $s=0$
		
		Create $\Phi$
		
		\While{$\texttt{successful\_seed} \leq S$}{
			Generate $\varepsilon_t(s) \sim N(0,\sigma^2I_p)$ in accordance with the seed $s$
			
			Compute the corresponding $Y_t(s) , t\in [T]$
			
			Use standard ADF test to test for $Y(s) \sim I(1)$. If failed at $5\%$ significance level, we reject the sample and re-draw by replacing $s$ by $s+1$.
			
			Run the appropriate algorithm and report $ \hat{r}(s)$, reaplce $\texttt{successful\_seed}$ by $\texttt{successful\_seed}+1$.
		}
		
	\end{algorithm}
	
\end{center}

\subsection{Summary of Results}

\subsubsection{Algorithm 1}
For algorithm 1, we simulate two sets of regime. In the first regime, we consider $p=10, T=100$, and we scan the cases for $r\in [5]$. In the second regime, we consider $p=30, T=300$, and we scan the cases for $r\in [5]$.

\autoref{3A} shows the percentage correct obtained under these regimes, plotted over the true $r$. For instance, when $r=0$, algorithm 1 detects 100\% correct in terms of $\hat{r}=0$ for both $p=10$ and $p=30$, whereas when $r=5$, the algorithm detects only 65\% correct in terms of  $\hat{r}=5$ for $p=30$, and 86\% correct for $p=10$.

\begin{figure}[H]
    \centering
    \includegraphics[width=0.8\columnwidth]{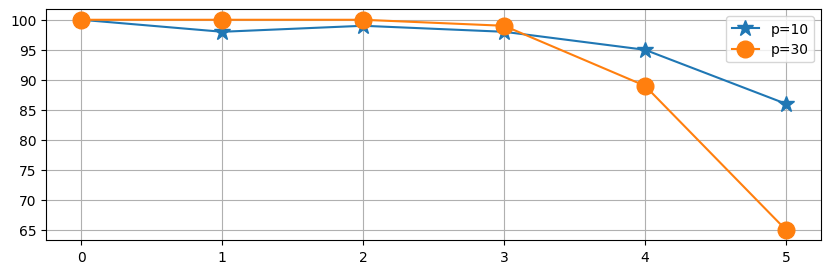}
    \caption{Percentage Correct (Accuracy of rank estimation)}
    \label{3A}
\end{figure}
\begin{figure}[H]
    \centering
    \includegraphics[width=0.8\columnwidth]{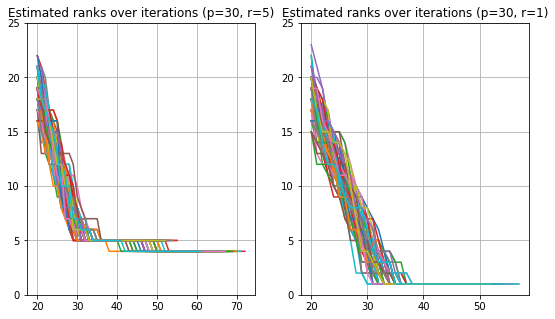}
        \caption{Estimated ranks over iterations of $\lambda_{0,j}$ increments}
    \label{3B}
\end{figure}

\autoref{3B} investigates potential issues arising from higher-rank under-penalisation. Particularly, we focus on the case of $p=30, r=5$ against the case of $p=30, r=1$. The latter got 100\% correct in accuracy, whereas the earlier performs poorly. As shown on the iteration graph on the left, while many terminates when $\hat{r}=5$, a good portion dropped to $\hat{r}=4$ and terminated there. This means the iterative increase of $\lambda_0$ has over-penalised the matrix, leading to a over-sparsified estimation of some rows --- leading to an under-estimation of ranks.

The promising result is on the right --- under more sparsified regimes, we can see all ranks converging to 1 and terminating there, which means no over-peanlisation by SSL specifications.

\subsubsection{Algorithm 2}
For algorithm 2, we demonstrate the advancement from algorithm 1 in the following two ways: better performance under the high-dimensional ($p=30, T=300$) regime, and the ability to handle even higher-dimensional ($p=100, T=500$) regime. In both regimes, we experiment  $r\in \{1, 5\}$. 

An overall summary for approximate percentage correct is illustrated in  \autoref{tab:T1}. The first column specifies the regime set-up and algorithm in comparison, whereas the second and third column give the definition of percentage correct. As indicated in the algorithms, the algorithm 2 only provide probabilistic estimation instead of scalar estimation of the ranks. Here, we take the average of all 100 different randomisation --- therefore many estimations may not be an integer, and hence the need to draw some notion of approximation. The second column of the table indicates percentage correct of falling into an estimation between $r-\frac{p}{100}$ and $r+\frac{p}{100}$. For instance, when $p=30,r=5$, algorithm 1 gives 65 estimations correctly within the 4.7-5.3 range, whereas algorithm 2 gives 94 estimations correctly within the same range. The algorithm 2 further estimate all scenarios correctly when we relax the range into between $r-\frac{p}{50}$ and $r+\frac{p}{50}$.

\begin{table}[h]
	\centering
	\begin{tabular}{c|c|c}
		$p,r$  &  $r \pm \frac{p}{100}$   & $r \pm \frac{p}{50}$ \\ \hline 
		$p=30, r=1$  (algorithms 1 and 2)   &  100 & 100 \\ \hline
		$p=30, r=5$  (algorithm 1)   &  65 & 65 \\ \hline
		$p=30, r=5$  (algorithm 2)   &  94 & 100 \\ \hline
		
		$p=100, r=1$  (algorithm 2)  &  84 & 100 \\ \hline
		$p=100, r=5$   (algorithm 2) &  100 & 100 \\ \hline
		
	\end{tabular}
	\caption{Approximate percentage correct by definitions}
	\label{tab:T1}
\end{table}
\begin{figure}[H]
	\centering
	\includegraphics[width=0.8\columnwidth]{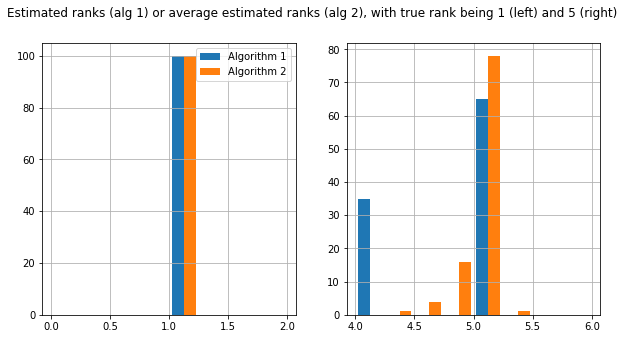}	
	\caption{Histograms of estimated ranks over samples for $p=30, T=300$}
	\label{3C1}
\end{figure}

As an additional comparison between algorithm 1 and 2, we show the histogram of estimated $\hat{r}$ across the 100 samples in \autoref{3C1} . The blue histograms indicate estimations given by algorithm 1 --- 100\% correctly in the case of $r=1$ whereas having a 35\% incorrect estimation ($\hat{r}=4$) when $r=5$. Algorithm 2 picks up the incorrectness in a more probabilistic manner --- hence the average being more concentrated around 5, resulting in more accurate results in \autoref{tab:T1}.
\begin{figure}[h]
	\centering
	\includegraphics[width=0.8\columnwidth]{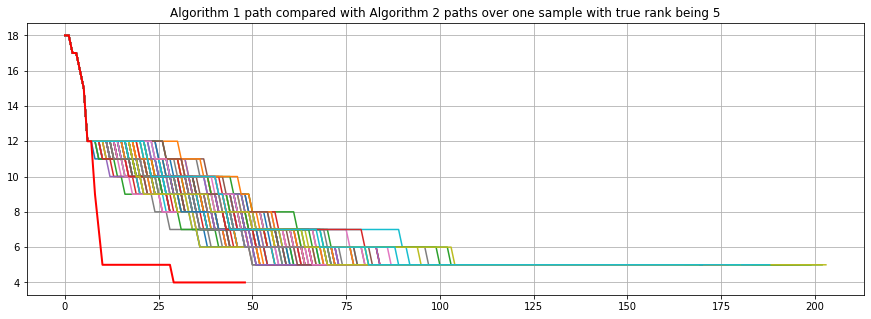}
	
	\includegraphics[width=0.8\columnwidth]{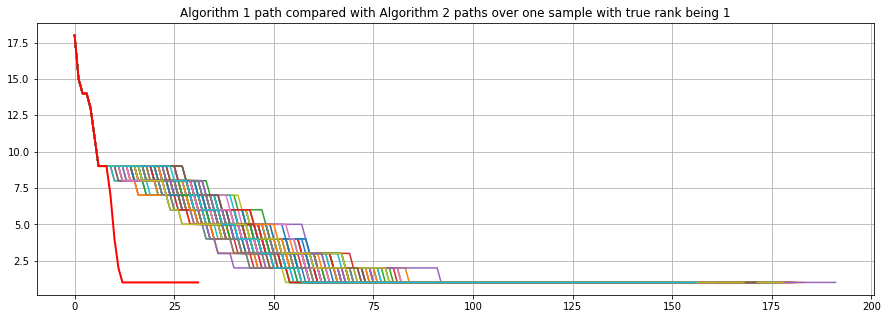}
	\caption{Estimated ranks over iterations of $\lambda_{0,j}$ increments}
	\label{3B2}
\end{figure}

\begin{figure}[h]
	\centering
	\includegraphics[width=0.8\columnwidth]{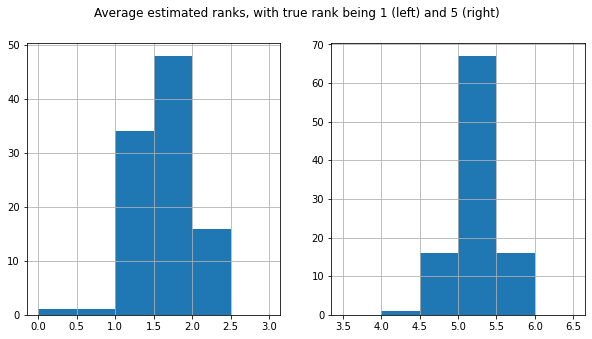}
	
	\caption{Histograms of estimated ranks over samples for $p=100, T=500$}
	\label{3C2}
\end{figure}

As an interactive showcase between the two algorithms, we draw one sample in the case of $r=5$ and $r=1$ and show the iteration paths from both algorithms in \autoref{3B2}. Algorithm 1 is highlighted in red. As per previous narratives, from the top graph, we can see the drop into $\hat{r}=4$ which concludes its fate into over-penalisation. For algorithm 2, the start of the iteration trajectories are exactly the same as algorithm 1 up to the iteration where $\hat{r}=12$. It then starts the randomised paths, resulting in potentially different pathways to draw down to, eventually, the level of $\hat{r}=5$ or $6$. Most of them eventually terminated at  $\hat{r}=5$. A similar story may be seen for the lower graph where termination occurred mostly at $\hat{r}=1$.

The cases where $p=100, T=500$ are presented in  \autoref{3C2} for histograms and  \autoref{3E} for one sample path. As can be seen, there are significant concentration at around the true rank for both cases --- more significantly so for $r=5$. This corroborates the summative statistics provided back in \autoref{tab:T1}.

In \autoref{3E}, we show an example of the case where $r=5$ and also in comparison with the red line (algorithm 1) which fails to estimate the true rank by under-estimating it to $\hat{r}=0$. In this example, mean estimation from algorithm 2 is 4.91 with median as 5.0.

\begin{figure}[h]
    \centering
\includegraphics[width=0.8\columnwidth]{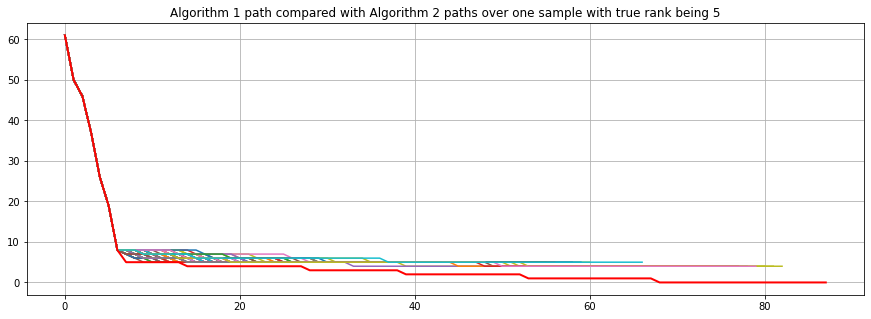}
    \caption{Estimated ranks over iterations of $\lambda_{0,j}$ increments}
    \label{3E}
\end{figure}

\section{Empirical application to stable portfolio}
The principle idea of this empirical application is to find a collection of unitary portfolio allocations $(\alpha_i)_i$ such that each portfolio $\alpha_i$ of assets has minimum variance --- hence denoted stable portfolio. By using cointegration analysis, it is expected that such minimum variance could persist not only in the training dataset, but also in the unseen testing dataset proceeding from the training. 

\subsection{Portfolio construction}
Recall in stock portfolio setting that $Y_t$ is a $p$-dimensional vector of stock price at time $t$, encompassing $p$ number of stocks. 
After obtaining  $\hat{R}$ and $\hat{r}$, the key idea is to use $\hat{R}$ to reconstruct $\hat{\Pi}$ and hence obtain $\hat{r}$ number of co-integrating vectors. These vectors, upon normalisation, yield exactly a portfolio which we may analyse, assuming a single dollar and zero-cost implementation on a long-short portfolio on the relevant universe of assets.

Let $\hat{\Pi}$ be the recovered $\Pi$ where $	\hat{\Pi} = (\tilde{S}\hat{R})^\intercal$.  Let $\hat{\Pi}_i$ be a non-zero row vector of $\hat{\Pi}$. A portfolio is therefore constructed with value at time $t$ being $\alpha_i Y_t$, where $\alpha_i := \frac{\hat{\Pi_i}}{||\hat{\Pi_i}||_1}$. The estimated rank $\hat{r}$ would henceforth be the number of such portfolios, as  $\alpha_1,..., \alpha_{\hat{r}}$ are linearly independent cointegration vectors.


\subsection{Data and Preparation}

We obtain daily stock prices of all individual stocks in SPX 500 as of end of September 2023. The price time series start from January 2022 to September 2023 (this gives 374 observations, though it also means $T<p$), and we use from the January 2022 to June 2023 for in-sample training, and July to September 2023 for out-of-sample testing. We use ADF test to check for I(1) non-stationarity. Out of the 500 stocks, 405 stocks are I(1). We are henceforth in a regime of $T=374$ and $p=405$.

The only pre-processing is to compute for cumulative return in relation to the start of the dataset, that is, to have $\frac{Y(t,j)}{Y(1,j)}$ as the time series for investigation.

\subsection{Main Results}

We run 100 repetitions of algorithm 2 to show various portfolio results, and then illustrate a number of implied portfolios from one of the results in algorithm 2 to demonstrate the performance and insights.

\subsubsection{Inference on estimated rank in algorithm 2}
Because of the probabilistic nature of \autoref{RwSSLEM2b}, we are able to produce slightly different $\hat{\Pi}$ upon different seeds for the uniform distribution, and hence potentially different $\hat{r}$ owe to different timings of termination of algorithms and different compositions of $\lambda_0$ vector. For a set of seeds $S^{seed}$, we are able to obtain $\{\hat{r}(s) : s \in S^{seed}\}$ which serves as an inference on $\hat{r}$.

In \autoref{fig:RI}, we plot the estimated ranks over iterations on the left, and the histogram of $\{\hat{r}(s) : s \in S\}$ on the right. An average rank of 1.7 is being estimated, as highlighted in the red verical line of the histogram.

\begin{figure}[h]
    \centering
    \includegraphics[width=0.47\columnwidth]{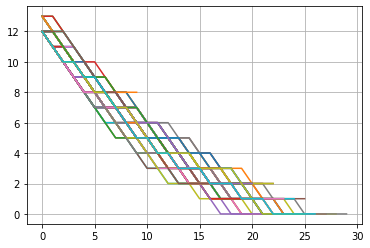}
      \includegraphics[width=0.47\columnwidth]{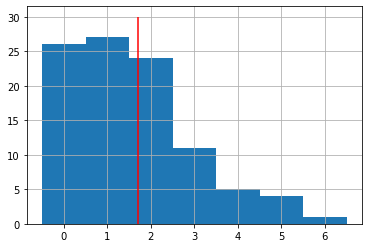}  
    \caption{Path of estimated ranks over iterations of $\lambda_{0,j}$ increments (left); Histogram of estimated ranks over different seeds (right)}
    \label{fig:RI}
\end{figure}

\subsubsection{General training and testing performance}
For each seed $s\in S^{seed}$, as long as $\hat{r}(s)>0$, we have  $\hat{r}(s)$ amount of portfolios implied from the cointegration relationships. We evaluate the stability of the portfolio by measuring the standard deviation of daily returns (in percentage), commonly known as `portfolio volatility' in financial literature.

In \autoref{fig:traintest}, we plot histograms of the training and testing volatilities of all such portfolios over all seeds, which gives 170 portfolios as the average rank is 1.7. The values of training and testing volatilities are much lower than the underlining index in both training and testing (see \autoref{portfolioperformance} for the index and equally-weighted portfolio volatility in the corresponding periods).

This has demonstrated the ability of cointegrated portfolio to reduce variance by a huge portion. 

\begin{figure}[h]
    \centering
    \includegraphics[width=0.7\columnwidth]{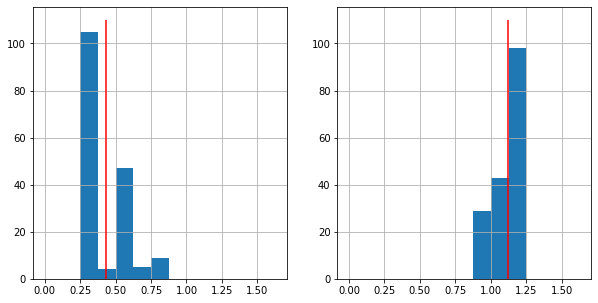}
    \caption{Portfolio volatility in training (left) and testing (right) over the 170 implied portfolios}
    \label{fig:traintest}
\end{figure}

\pagebreak

\subsubsection{Example portfolio performance  and insight}
\begin{table}[H]
	\centering
	\begin{tabular}{|l|r|r|}
		\hline
		Portfolio &  Training & Testing \\ \hline
		SPY Index &   18.4 &   10.8       \\ \hline 
		Equally Weighted & 18.2 &  10.8    \\ \hline 
		Cointegration 1 &  0.4 &  1.1  \\  \hline  
		Cointegration 2 &  0.3 &  1.1   \\  \hline  
		Cointegration 3 &  0.9 &  1.0   \\  \hline  
		
	\end{tabular}
	\caption{Training and Testing period standard deviation of daily returns (in percentage)}
	\label{portfolioperformance}
\end{table}

We select a particular seed where $\hat{r}(s)=3$ and illustrate the results from the corresponding 3 portfolios. In \autoref{portfolioperformance}, we report the portfolio volatilities of the three portfolios --- there is a massive decrease of volatilities from the underlining index / equally weighted portfolio. We plot the cumulative returns in \autoref{port-cr} for training period (upper) and testing period (lower), and there is a clear and meaningful stability in contrast to the index and equally weighted portfolio.

\begin{figure}[h]
    \centering
    \includegraphics[width=\columnwidth]{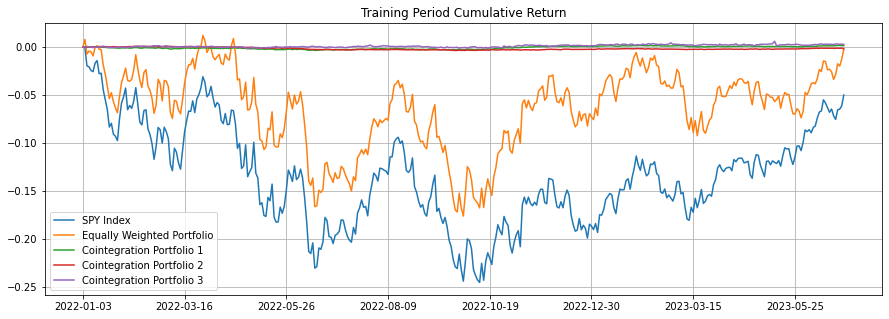}
    \includegraphics[width=\columnwidth]{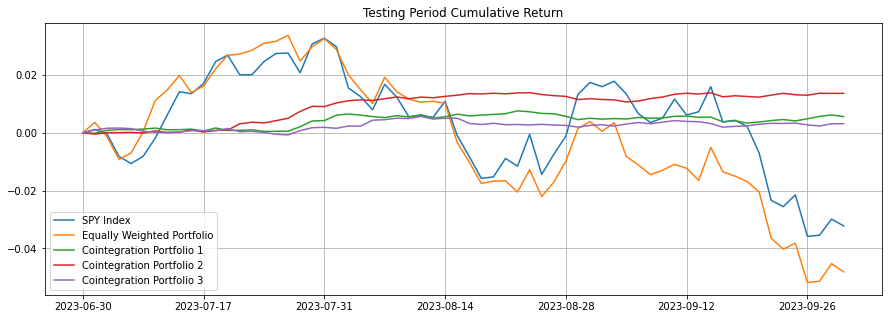}
    \caption{Cumulative return plots.}
    \label{port-cr}
\end{figure}

As to the portfolio insight, we investigate the top 25 stocks' weighting of each of the 3 portfolios in the first three graphs of \autoref{port-ins}. These weights are also plotted against the SPY index and horizontal lines indicating the equally-weighted weights. We can see that though some handful similarities, the main long-short focus of the portfolios differ significantly to the SPY index. An alternative graph at the bottom illustrates the weighting of the top 25 stocks from SPY in the respective portfolios.This also supports the insight that the focus of the cointegrated portfolios has not been on the highly capitalised stocks, compared to SPY.

\begin{figure}[p]
    \centering
    \includegraphics[width=\columnwidth]{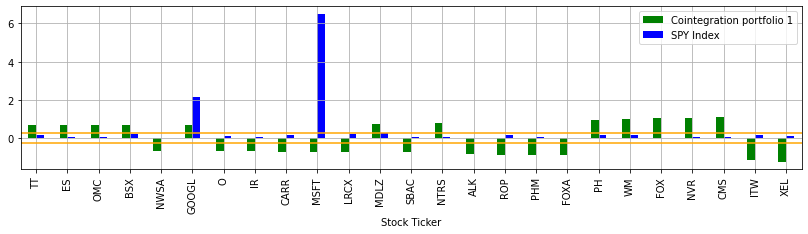}
     \includegraphics[width=\columnwidth]{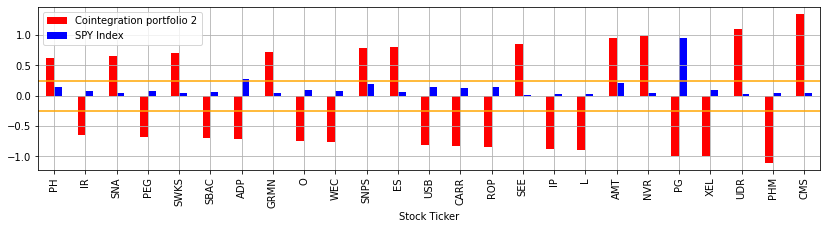}
        \includegraphics[width=\columnwidth]{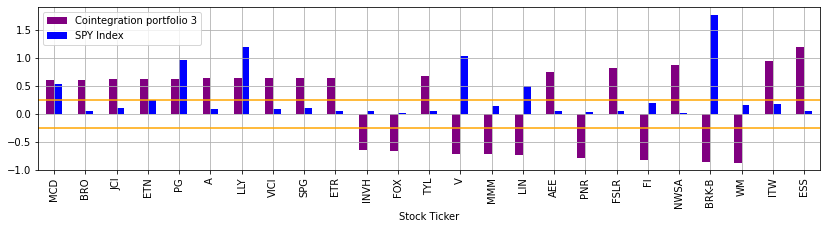}
           \includegraphics[width=\columnwidth]{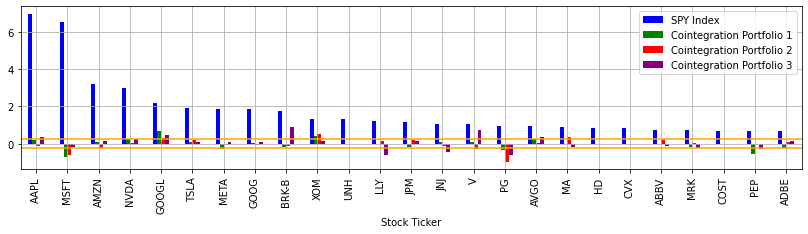}
    \caption{Portfolio Insights.}
    \label{port-ins}
\end{figure}

\pagebreak

\section{Conclusion}

\subsection{Discussions}
Though algorithm 2 has the drawbacks of requiring some moderate amount of sparsity assumption, it significantly outperforms algorithm 1 by halting the $\lambda_0$ increment appropriately. Moreover, algorithm 2 enables some probabilistic inference on the estimation $\hat{r}$, in the sense that $j \sim U(P)$, i.e. $\{\lambda_{0,j}: j \in [P]\}$ being attributed on some randomised size each time the iteration goes. This adds a further Bayesian layer in the modelling hierarchy, and gives an interpretable histogram as one designs the number of repetitions to run the algorithm.

In terms of the higher dimensional investigation, we reached to $p=100, T=500$ in simulation and $p=405, T=374$ in empirical application. It is worth appreciating that algorithm 2 scales down computing time for large $p$, as it only requires one $\beta_j$ update per iteration due to the sampled index, instead of mass update of the whole matrix as per algorithm 1. A compute-time relevant adjustment is also available at line 3 --- in essence, one may wish to tune down the sparsity assumption (such as subject to an identifiability constraint $pr<T$) and therefore increase $\Delta_\lambda$ to achieve faster runtime. This is not possible to be done in algorithm 1.

A number of future works are thought to be helpful based on the current investigation: we may investigate a even larger number of features, either in the numerical sense (e.g. $p>1000$) or in the asymptotic sense of $p = O(T)$. The asymptotic inference may also be drawn towards approximate correctness, such as  $\hat{r}/r \to 1$ with more generic limits such as $q, r\to\infty$. This then leads to space for investigating even higher dimensional cointegration in applied fields, such as portfolio construction and adjustments. It is also worth extending the SSL-VECM Bayesian model on covariance matrices, which may play a role in some cointegrated relationships. Though, such covariance may be sparse, adding another level of difficulty for the model estimation.

\subsection{Contribution}

There are mainly three forces of contributions being brought from this paper. First, similar to \textcite{LiangSchienle19} and \textcite{ChenSchienle22}, we extend the `traditional identifiability constraint' from $p^2<T$ to $pr \leq T$, in the sense that through penalisation, one does not necessarily need more observations than unknowns to estimate the cointegration matrix and thereafter the number of cointegration relationships in the system. Instead, through SSL, we proposed nuanced methodologies to estimate $r$ and also showed some level of asymptotic convergence in a Bayesian setting.

Secondly, we used various simulation regimes to test our proposed algorithms. This furthers the SSL executions from optimally finding a $p$ dimensional sparse vector to finding a $p\times p$ dimensional sparse matrix and determining its true rank. Particularly, the termination of $\lambda_0$ increment varies substantially from existing Bayesian literature, as we quantify further the termination criteria instead of a generic guidance on convergence-termination. This contributes to the implementation of SSL and furthered the agenda to use Bayesian modelling as a `pragmatic frequentist' on penalisation.

Lastly, the application to portfolio construction contributes to the relevant financial literature. Our algorithms yield various independent and insightful portfolios that yield low volatilities, as cointegration theory would presume, which invite further studies on topics such as efficient heding and portfolio construction.

\pagebreak

\appendix

\section{Technical appendix}

\subsection{Data processing}\label{DP}

Recall that our data, in matrix form, can be written as $A=[\Delta Y_1, ..., \Delta Y_T] = [A_1,...,A_T] \in \mR^{p \times T}$ and $B = [\tilde{S}^T Y_0,..., \tilde{S}^T Y_{T-1}] = [B_1,...,B_T]\in \mR^{p \times T}$ where $A_t, B_t \in \mR^{p \times 1}$. The data processing tasks concerns transformation from $A$ to $\tilde{A}$ and $B$ to $\tilde{B}$ such that
\begin{itemize}
    \item $\tilde{A}$ is centralised, i.e. $\sum_{t\in[T]} \tilde{A_t} = \boldmath{0}_p$. This can be done by \begin{equation}
        \tilde{A}_t = A_t - \frac{\sum_{t\in[T]} A_t}{T}
    \end{equation}
    \item  $\tilde{B}$ is centralised and with $\sqrt{T}$ of standard deviation each row. This can be done by \begin{equation}
        \tilde{B}_{t,j} = \sqrt{T} \frac{B_{t,j}-\mu_{B,j}}{\sigma_{B,j}} \text{ where } \mu_{B,j} = \frac{\sum_{t\in[T]} B_{t,j}}{T} \text{ and }
        \sigma_{B,j} = \sqrt{\sum_{t\in [T]} \frac{(B_{t,j} - \mu_{B,j})^2}{T} }
    \end{equation}
    
\end{itemize}

\subsection{Derivation to \autoref{R_BE}}\label{DBE}
First, we note $\pi(\phi|D, D_0) \propto \pi(D|\phi, D_0) \pi(\phi)$ where $\pi(D|\phi, D_0) $ is the likelihood conditional on some initialisation $D_0$ 
and $\pi(\phi)$ is the prior of $\phi$. $\propto$ refers to proportional up to some constants not involving $\phi$. 

Now, we can write \begin{align}
\pi(\phi|D, D_0) &\propto \pi(D|R, \sigma^2, D_0) \pi(\sigma^2) \pi(R) \\
  &= \left( \prod_{t\in[T]} \pi(	\tilde{A}_{t}|\tilde{B}_t, R, \sigma^2) \right ) \pi(\sigma^2) \left(  \prod_{j\in[p]} \sum_{\gamma_j \in \{0,1\}} \pi(R(\cdot, j) | \gamma_j) \pi(\gamma_j | \theta_j)  \right)
\end{align}

We further observe, for any $t\in [T]$, $j\in[p]$, \begin{align}
	\pi(	\tilde{A}_{t}|\tilde{B}_t, R, \sigma^2) &\propto \sigma^{-p} \exp\left(-\frac{||e_t(R)||_2^2}{2\sigma^2} \right)  \\
	\sum_{\gamma_j \in \{0,1\}} \pi(R(\cdot, j) | \gamma_j) \pi(\gamma_j | \theta_j) & = (1-\theta_j)\psi_{0,j}(R(\cdot, j)) + \theta_j  \psi_1(R(\cdot, j)) 
	\\&\propto \exp(-\lambda_1 ||R(\cdot,j)||_1 ) \times \prod_{i\in[p]} \frac{1}{p_j^*(R(i,j),\theta_j)}
\end{align}

Henceforth combining everything, we have \begin{equation*}
	\pi(\phi|D, D_0) \propto\left( \prod_{t\in[T]}  \exp\left(-\frac{||e_t(R)||_2^2}{2\sigma^2} \right)\right) \sigma^{-(pT+2)}  \prod_{j\in[p]}  \left(  \exp(-\lambda_1 ||R(\cdot,j)||_1 ) \times \prod_{i\in[p]} \frac{1}{p_j^*(R(i,j),\theta_j)}  \right)
\end{equation*}

Then we take logarithm and obtain 
\begin{equation}
		(\hat{R}^{BE}, \hat{\sigma}^{BE}) = \argmin_{R\in \mR^{p\times p} 
			, \sigma \in \mR_{> 0}} - \log(\pi(\phi|D, D_0) )
\end{equation}
to the desired \autoref{R_BE}.
\subsection{Derivation to \autoref{R_BE2}}\label{DBE2}
Firstly, breaking down we obtain
 \begin{align}
	\pi(\phi|D, D_0) &\propto 
	 \left( \prod_{t\in[T]} \prod_{j\in[p]} \pi(	\tilde{A}_{t,j}|\tilde{B}_t, R, \sigma_j^2) \right) \left( \prod_{j\in[p]} \pi(\sigma_j^2)  \right)   \left( \prod_{j\in[p]} \left( \int_{(0,1)} \pi(R(\cdot,j)|\theta_j)\pi(\theta_j) d\theta_j  \right) \right)
\end{align}
Similar to Appendix \ref{DBE}, we proceed to observe, for any $t\in [T]$, $j\in[p]$,
\begin{align}
	\pi(\tilde{A}_{t,j}|\tilde{B}_t, R, \sigma_j^2) &\propto \sigma_j^{-1} \exp\left(-\frac{(e_t(R))_j^2}{2\sigma_j^2} \right)  \\
	&=\sigma_j^{-1} \exp(-\epsilon(j)^2_t) \\
	\pi(R(\cdot,j)|\theta_j)\pi(\theta_j)  & = ((1-\theta_j)\psi_{0,j}(R(\cdot, j)) + \theta_j  \psi_1(R(\cdot, j))) \pi(\theta_j) 
	\\&= \exp(-\lambda_1 ||R(\cdot,j)||_1 ) \times \prod_{i\in[p]} \frac{\theta_j\pi(\theta_j)}{p_j^*(R(i,j),\theta_j)}
\end{align}
Therefore 
 \begin{align}
	\pi(\phi|D, D_0) &\propto 
	\left( \prod_{j\in[p]} \exp(-||\epsilon(j)||_2^{2}) \sigma_j^{-(T+2)} \right)  \left( \prod_{j\in[p]} \left( \exp(-\lambda_1 ||R(\cdot,j)||_1 ) \times \int_{(0,1)}   \prod_{i\in[p]} \frac{\theta_j\pi(\theta_j)}{p_j^*(R(i,j),\theta_j)} d\theta_j  \right) \right)
\end{align}

Then we take logarithm and obtain 
\begin{equation}
		(\hat{R}^{BE}, (\hat{\sigma}_j^{BE})_{j\in[p]}) = \argmin_{R\in \mR^{p\times p} 
		, ({\sigma}_j^{BE})_{j\in[p]} \in \mR^p_{> 0}} - \log(\pi(\phi|D, D_0) )
\end{equation}
to the desired \autoref{R_BE2}.

\pagebreak
\printbibliography
\end{document}